\documentclass[letterpaper,11pt]{article}

\usepackage[utf8]{inputenc} 
\usepackage[T1]{fontenc}

\usepackage{amsmath, amssymb, amsthm, mathtools, bbm, nicefrac}

\usepackage[noend]{algorithmic}
\usepackage[linesnumbered,ruled,vlined]{algorithm2e}

\usepackage{graphicx}
\usepackage{caption}
\usepackage{subcaption}
\usepackage{tikz, pgfplots}
\usetikzlibrary{arrows.meta,shapes,automata,calc,backgrounds,petri,snakes,
positioning, decorations.pathreplacing}
\usepackage[draft]{changes} 
\definechangesauthor[name={Masoud}, color=blue]{ms}

\usepackage{fullpage}
\usepackage[margin=1in]{geometry}
\usepackage{setspace}
\usepackage{xspace}
\usepackage{color}
\definecolor{crimsonglory}{rgb}{0,0,0} 
\usepackage[normalem]{ulem} 

\usepackage{enumitem}
\usepackage{threeparttable}
\usepackage{verbatim}

\usepackage[colorlinks=true, linkcolor=blue, citecolor=blue, urlcolor=blue]{hyperref}

\theoremstyle{plain}
\newtheorem{theorem}{Theorem}[section]

\DeclarePairedDelimiter\floor{\lfloor}{\floor}

\usepackage{thm-restate}
\usepackage{cleveref}
\crefname{LP}{LP}{LPs}

\newcommand{\ignore}[1]{}

\SetCommentSty{mycommfont}
\SetKwInput{KwInput}{Input}
\SetKwInput{KwOutput}{Output}

\newcounter{proccnt}

\makeatletter
\def\GrabProofArgument[#1]{ #1: \egroup\ignorespaces}
\def\proof{\noindent\textbf\bgroup Proof%
	\@ifnextchar[{\GrabProofArgument}{. \egroup\ignorespaces}}

\makeatother

\title{Lower Bound for Online \MMS\ Assignment of Indivisible Chores}
\author{Masoud Seddighin \and Saeed Seddighin}
\date{}

\newcommand{\MMS}{\textsf{MMS}}

\newcommand{\sce}{\textsf{sc}}

\newcommand{\cost}{V}
\newcommand{\valu}{V}
\newcommand{\chores}{M}
\newcommand{\agents}{N}

\newcommand{\agent}{a}

\newcommand{\chore}{b}

\begin{document}
	\maketitle
	\thispagestyle{empty}
	
	\begin{abstract}
		We consider the problem of online assignment of indivisible chores under \MMS\ criteria. The previous work proves that any deterministic online algorithm for chore division has a competitive ratio of at least 2. In this work, we improve this bound by showing that no deterministic online algorithm can obtain a competitive ratio better than $n$ for $n$ agents.
	\end{abstract}
	\section{Introduction}\label{introduction}

Fair division is a well-established problem in mathematics and economics, with broad relevance to disciplines such as political science, social science, and computer science \cite{Dubins:first,Steinhaus:first,brams1996fair}. The concept was first formally introduced by Hugo Steinhaus under the framework of "cake cutting," which focuses on distributing a divisible resource—the "cake"—among multiple agents with varying preferences. For over eight decades, this problem has attracted significant scholarly interest \cite{even1984note,stromquist1980cut,brams1996fair}. In more recent years, attention has shifted toward a discrete variant of the problem involving the allocation of indivisible tasks or goods, especially within computer science \cite{Procaccia:first,ghodsi2018fair,caragiannis2016unreasonable,amanatidis2015approximation,Budish:first,Saberi:first,robertson1998cake}. Rather than dividing a divisible resource such as a cake, the goal here is to allocate a set of indivisible items or tasks among 
$n$ agents. In such cases, direct proportional division is often impossible—for example, when only a single item must be assigned.

This work focuses on the dual problem: the fair assignment of chores. Unlike goods, which are desirable, chores represent undesirable tasks that impose costs on the agents. The objective is to distribute these chores in a way that no agent bears more than their fair share of the burden. Our focus lies specifically on the indivisible setting, using the maximin (\MMS) fairness criterion introduced by Budish~\cite{Budish:first} to guide the allocation which has been used in many works~\cite{farhadi2017fair,maseed123,akrami2024randomized,our-ec-paper,akrami2023simplification,Procaccia:first,Procaccia:second,kurokawa2015can}.

Formally, let $\agents$ be a set of $n$ agents, and $\chores$ be a set of $m$ chores. We denote the agents by $\agents=\{\agent_1, \agent_2, \ldots, \agent_n\}$ and chores by $\chores = \{\chore_1, \chore_2,\ldots,\chore_m\}$. Each agent $\agent_i$ has an \textit{additive} function $\cost_i$ over the chores which denotes her cost for each subset of chores. Let $\Pi(\chores)$ denote the set of all $n$-partitionings of the chores. The maximin  value of agent $\agent_i$ (denoted by $\MMS_i$) is defined as

\[
\MMS_i = \max_{\langle \pi_1, \pi_2, \ldots, \pi_n \rangle \in \Pi} \min_{1 \leq j \leq n} \valu_i(\pi_j),
\]
where $\Pi$ represents the set of all possible partitionings of $\chores$ into $n$ bundles, and $\valu_i(\pi_j)$ is the cost that agent $\agent_i$ assigns to bundle $\pi_j$. An allocation is called \textsf{MMS} (or $\beta$-\textsf{MMS}) if every agent $\agent_i$ receives a bundle which costs at most $\MMS_i$ (or $\beta \cdot \MMS_i$) to her.

In this setting, the goal is to assign chores to agents such that the cost each agent incurs from their assigned bundle is within a constant multiplicative factor of their $\MMS$ value. The $\MMS$ value acts as a natural lower bound, since even in the case of identical costs across agents, at least one agent must receive a bundle with cost no less than their $\MMS$. This leads to the central question: \textit{Can we find chore assignments where each agent's total cost is bounded by a constant multiple of their $\MMS$ value?} The current best known upper bound for this factor is $\frac{13}{11}$~\cite{huang2023reduction}.

A natural and practical extension of the fair chore division problem is its online variant, where chores arrive sequentially over time rather than being known in advance. In this setting, we assume that the number of agents is known ahead of time, but the chores—and the agents’ costs for them—are revealed one by one as they arrive. Upon the arrival of each chore, we learn how costly it is for each agent and must assign it immediately and irrevocably, without knowledge of future chores. This online formulation captures real-world scenarios such as task assignment in dynamic work environments, where decisions must be made in real time. The central challenge is to design an online algorithm that ensures fairness—measured, for instance, by the maximin share (\MMS) benchmark—despite the uncertainty and lack of full information.

Previous work~\cite{zhou2023multi} shows that for online allocation of indivisible goods, no fraction of the \MMS\ benchmark can be guaranteed to all agents via a deterministic algorithm. For the online assignment of indivisible chores, the best known lower bound for deterministic algorithms is a lower bound of $2$~\cite{zhou2023multi}. In other words, it is proven that for any deterministic online algorithm for assignment of indivisible chores there is one scenario in which one agent receives a bundle of chores that is twice as costly as her \MMS\ value. In this work, we improve this bound to a lower bound of $n$ for $n$ agents.

	\section{Preliminaries}
Let $\agents$ be a set of $n$ agents, and $\chores$ be a set of $m$ chores. We denote the agents by $\agents=\{\agent_1, \agent_2, \ldots, \agent_n\}$ and chores by $\chores = \{\chore_1, \chore_2,\ldots,\chore_m\}$. Each agent $\agent_i$ has an \textit{additive} function $\cost_i$ over the chores which denotes her cost for each subset of chores.  As proposed by Budish~\cite{Budish:first}, the maximin value of agent $\agent_i$ (denoted by $\MMS_i$) is defined as
\begin{equation}\label{eq:1}
	\MMS_i = \min_{\langle A_1, A_2, \ldots, A_n\rangle \in \Pi(\chores)} \max_{\agent_j \in \agents} \cost_i(A_j).
\end{equation}
We define an assignment as being \emph{$\alpha$-\MMS} if all chores are assigned to agents and each agent $\agent_i$ receives a bundle $S_i$ of chores such that $\cost_i(S_i) \leq \alpha \cdot \MMS_i$.

In our online setting, initially, our algorithm is aware of the value of $n$. It is important to note that our algorithm is not even aware of the value of $m$ in the beginning. The chores then arrive one by one, and each time a chore arrives our algorithm learns the cost of each chore for each agent. After that, it has to make an irrevocable decision as to which agent the chore is assigned to. The algorithm terminates after all chores arrive. We call an algorithm $\alpha$-competitive, if it assigns the chores to the agents in a way that after the termination of the algorithm, the assignment is $\alpha$-\MMS.

	\section{Lower Bound of $n$ for Deterministic Algorithms}
It has been proven in previous work~\cite{zhou2023multi} that any deterministic online algorithm for chore division has a competitive ratio of at least $2$. Indeed, the trivial algorithm that assigns all chores to an agent has a competitive ratio of $n$. We improve the lower bound of $2$ to a lower bound of $n$ and close this gap by showing that no deterministic algorithm for online chore division has a competitive ratio better than $n$.

\begin{theorem}
	No deterministic algorithm for online chore division has competitive ratio better than $n$ for $n$ agents.
\end{theorem} 
\begin{proof}
For any $\epsilon > 0$, we construct an online scenario, namely $\sce(n,\epsilon)$, for which any deterministic algorithm fails to be $(n-n\epsilon)$-competitive. We construct our scenario in a recursive manner. We parametrize our scenarios by two quantities. More precisely, we represent a scenario by $\sce(i,\epsilon)$ where the algorithm is guaranteed to only assign chores to agents $\{\agent_1, \agent_2, \ldots, \agent_i\}$ and that at the end of the scenario we can be sure that one of the agents has received a bundle of chores whose cost for her is more than $(n-n\epsilon)$ times her $\MMS$\ value. Each scenario assumes that an arbitrary sequence of chores have arrived so far and the algorithm has assigned them in an arbitrary manner and from then on the scenario decides how the chores are created. Scenario $\sce(i,\epsilon)$ adaptively constructs the next chores and guarantees that if the algorithm is limited to assign chores to only $i$ agents, then the assignment fails to be $(n-n\epsilon)$-$\MMS$ after the scenario ends.

 We begin by considering a scenario wherein we are guaranteed that the algorithm will only assign chores to agent $\agent_1$. We use that scenario to construct a scenario wherein the algorithm is guaranteed to only assign the chores to agents $\agent_1$ or $\agent_2$. More generally, each time we use the scenario that traps the algorithm under the assumption that the algorithm has to assign the chores to only $i$ agents, and construct a similar scenario under the assumption that the algorithm has to assign chores to only $i+1$ agents. The final scenario we construct assumes that the algorithm has to assign each of the chores to one of the $n$ agents which aligns with our problem setting. 
 
Let us start with the simplest case of $\sce(1,\epsilon)$ where the algorithm has to only assign chores to agent $\agent_1$. In this case, let $x$ be the sum of the costs of all the previous chores that have arrived for agent $\agent_1$. From here on, we define the next $n$ chores to have a cost of $w=x/\epsilon$ for agent $\agent_1$. Notice that since the algorithm is guaranteed to assign all the coming chores to agent $\agent_1$, the cost of the chores for the rest of the agents do not matter in our analysis. It follows that after assigning these $n$ chores to agent $\agent_1$, the total cost of agent $\agent_1$ for the chores assigned to her is at least $nw$ and her $\MMS$\ value is at most $w+x \leq w(1+\epsilon)$. Thus, if we terminate the algorithm after those $n$ chores, the competitive ratio of the algorithm would be at least 
\begin{align*}
\frac{nw}{w(1+\epsilon)} &= \frac{n}{1+\epsilon} \\
&> (1-\epsilon)n.
\end{align*}

Now, we use this to construct $\sce(2,\epsilon)$. Again, we assume that some chores have arrived so far and the algorithm has assigned them to the agents in an arbitrary manner. We define $\epsilon'=\epsilon/4$ and $x$ to be the sum of the costs of the chores that have arrived so far for agent $\agent_1$. Also, we define $\ell_1 = n$ to be the maximum number of chores that arrive in scenario $\sce(1,\epsilon)$ before the algorithm terminates. We define $w = x/\epsilon'$ and keep introducing chores by repeating the following pattern consecutively. In this pattern, we keep adding chores whose costs for agent $\agent_1$ increase exponentially until one of those chores is assigned to agent $\agent_1$. At that point, the pattern ends and we start over by repeating the pattern from the beginning. In our pattern we introduce the chores in the following way:
\begin{itemize}
	\item The cost of the first chore for agent $\agent_1$ is $w/(1+1/\epsilon')^{\ell_1}$. The cost of the second chore for agent $\agent_1$ is $w/(1+1/\epsilon')^{\ell_1-1}$ and so on. 
	\item The costs of the next $\ell_1$ chores for agent $\agent_2$ is determined based on $\sce(1,\epsilon)$ until this pattern ends.
\end{itemize}
It follows that if $\ell_1$ consecutive chores are given to agent $\agent_2$ then the algorithm would be not be $(n-n\epsilon)$-competitive by the guarantees of $\sce(1,\epsilon)$. Thus, at some point one of the chores should be given to agent $\agent_1$. Let this be the $i$'th chore in this sequence. Since the chores have increasing costs for agent $\agent_1$ in an exponential way we have:
$$w/(1+1/\epsilon')^{\ell_1-i+1} \geq (1/\epsilon') \sum_{1 \leq j < i} w/(1+1/\epsilon')^{\ell_1-j+1}.$$
In other words, the cost of the chore given to agent $\agent_1$ for her is at least $1/\epsilon'$ times the cost of the chores that are not given to agent $\agent_1$ for her. Once this chore is given to agent $\agent_1$, we start over the next pattern. This repetition stops when after the end of a pattern, the total cost of the chores for agent $\agent_1$ generated in this process is at least $nw$. It follows that the cost of the chores given to agent $\agent_1$ in this process is at least 
\begin{align*}
\frac{1/\epsilon'}{1+1/\epsilon'} nw \geq (1-\epsilon') nw.
\end{align*}
By the way we construct the chores, the total cost of all chores for agent $\agent_1$ would be bounded by 
\begin{align*}
x + nw + \sum_{j = 1}^{\ell_1-1} w/(1+1/\epsilon')^{\ell_1-j+1} &= \epsilon'w + nw  + \sum_{j = 1}^{\ell_1-1} w/(1+1/\epsilon')^{\ell_1-j+1}\\
& \leq \epsilon'w + nw + \epsilon' w\\
&  = (n+2\epsilon') w
\end{align*}
Moreover, the largest cost of a chore for agent $\agent_1$ is bounded by $w/(1+1/\epsilon') \leq \epsilon'w$ 
 and thus $$\MMS_1 \leq nw/n + 3\epsilon' w = (1+3\epsilon')w.$$ This yields a competitive ratio at least 
\begin{align*}
	\frac{(1-\epsilon')nw}{(1+3\epsilon')w}  &= n\frac{(1-\epsilon')}{(1+3\epsilon')}  \\
	&\geq n(1-4\epsilon') \\
	& \geq n(1-\epsilon).
\end{align*}

The transition from each $\sce(i,\epsilon)$ to $\sce(i+1,\epsilon)$ follows the exact same blueprint except that the values of $\ell_i$'s increase exponentially. For example $\ell_1 = n$ is polynomial in terms of $n$ but the value of $\ell_2$ is exponentially large.  For example, the cost of the smallest chore that we generate for $\sce(2,\epsilon)$ would be $w/(1+1/\epsilon')^{\ell_1}$ and thus it takes at most
\begin{align*}
\ell_2 \leq \ell_1 + \frac{nw}{w/(1+1/\epsilon')^{\ell_1}} &= \ell_1 + n(1+1/\epsilon')^{\ell_1} 
\end{align*}
steps to terminate $\sce(2,\epsilon)$.
\end{proof}

	\bibliographystyle{plain}
	\bibliography{draft}
	\appendix
	
\end{document}